\documentclass[11pt]{article}
\usepackage{fullpage,amssymb,amsmath,amsthm,graphicx,natbib,authblk}

\usepackage[hyperfootnotes=false]{hyperref}

\renewcommand{\v}[1]{#1}
\newcommand{\m}[1]{#1}

\newcommand{\Exp}[1]{{\rm E}[ \ensuremath{ #1 } ]  }
\newcommand{\Var}[1]{{\rm Var}[ \ensuremath{ #1 } ]  }

\newcommand{\Sig}{{\v \Sigma}}

\long\def\symbolfootnote[#1]#2{\begingroup%
\def\thefootnote{\fnsymbol{footnote}}\footnote[#1]{#2}\endgroup}

\title{Marginally Specified Priors for Nonparametric Bayesian Estimation}
\author[1]{David C. Kessler}
\author[2]{Peter D. Hoff}
\author[3]{David B. Dunson}
\affil[1]{Department of Biostatistics, University of North Carolina, Chapel Hill}
\affil[2]{Departments of Statistics and Biostatistics, University of Washington}
\affil[3]{Department of Statistical Science, Duke University}

\date{\today}

\begin{document}
\maketitle

\symbolfootnote[0]{Address for correspondence: \url{pdhoff@uw.edu}. 
David Kessler's work was partially supported by NIEHS
training grant T32ES007018.
Peter Hoff's work was partially supported by
NICHD grant 1R01HD067509-01A1. }

\begin{abstract}

Prior specification for nonparametric Bayesian inference
involves the difficult task of quantifying prior knowledge about a parameter of high,  often
infinite, dimension.
Realistically, a statistician is unlikely to have informed opinions
about all aspects of such a parameter, but may have real information about functionals
of the parameter, such the population mean or variance.
This article proposes a new framework for nonparametric Bayes inference
in which  the prior  distribution
for a possibly infinite-dimensional parameter
is decomposed into two parts:
an informative prior on
a finite set of functionals, 
and a
nonparametric conditional prior for the parameter given the functionals.
Such priors can be easily constructed from  standard
nonparametric prior distributions in common use, and
inherit the large support of the
standard priors upon which they are based.
Additionally, posterior approximations under these informative priors
can generally be made via minor adjustments to existing Markov chain
approximation algorithms for standard nonparametric prior distributions.
We illustrate the use of such priors in the context of multivariate density estimation using Dirichlet process mixture models, and in the  modeling of high-dimensional sparse contingency tables.

\medskip


\end{abstract}

\section{Introduction}

Many real-world data analysis situations do not lend themselves well to
simple statistical models indexed by a finite-dimensional parameter.
This has led to the development of a rich class of nonparametric Bayesian  (NP Bayes) methods,
the general idea of  which is to
obtain inference under a prior that has support on the entire space of relevant probability distributions \citep{ferguson_1973}.
These methods have been applied to a variety of problems, such as density estimation \citep{mulleretal_1996}, image segmentation \citep{suddjord2008}, speaker diarization \citep{foxetal_2011}, regression and classification \citep{neal_1998}, functional data analysis \citep{petroneetal_2009} and quantitative trait loci mapping \citep{zouetal_2010} to name only a few.
This breadth of applications reflects  the utility of NP Bayes methods in
modern statistical data analysis.

Many NP Bayes methods are built upon 
either the Dirichlet distribution (DD) for finite sample spaces or the 
Dirichlet process (DP)   \citep{ferguson_1973} for infinite 
sample spaces. 
For the latter case, 
the body of work on parameter estimation \citep{escobar_1994}, density estimation and inference \citep{escowest_1995} and the steady improvement in sampling methods
\citep{escobar_1994, walker_2007, yau_etal_2011, kallislice_2011} have all made the DP prior an attractive choice for many applications.
For a given sample space $\mathcal Y$,
a DD or  DP prior over distributions on $\mathcal Y$
is parameterized in terms of a ``base measure'' $Q_0$ on $\mathcal Y$  and
a ``concentration parameter'' $\alpha$.
Although samples from the DP prior are discrete with probability one,
this prior
is nonparametric in the sense that it
 has
weak support on the set of all distributions having the same support as
$Q_0$. 
Analogously, the DD prior is nonparametric in the sense that it has  support on the
entire $(|\mathcal  Y|-1)$-dimensional simplex. 
For both the DD and DP, 
a large value of $\alpha$ corresponds to a prior
concentrated near $Q_0$. For the DP, a small $\alpha$
results in distributions with probability mass concentrated on only a few points, drawn
independently from $Q_0$.
For the DD, a small $\alpha$ can result in mass being concentrated near the 
vertices of the simplex. 

For many NP Bayes methods, the DP is used as a prior
for a mixing distribution in a mixture model:
The data are assumed to come from a population with density
$p(y|Q ) = \int p(y|\psi) Q(d\psi)$, where $\{p(y|\psi) :  \psi\in \Psi\}$
is a simple parametric family.  A DP prior on $Q$ results in a
Dirichlet process mixture model (DPMM)
\citep{lo_1984,escowest_1995, maceachern_muller_1998}. As $Q$ is discrete with
probability 1, the resulting model for the population distribution
is a countably infinite mixture model, where
the parameters in the component measures  are determined by $Q_0$, and
the number of components
with non-negligible weights is increasing in $\alpha$.

Clearly, the choice of $\alpha$ and  $Q_0$ will have a significant effect on  the prior for the population density, and potentially on posterior inference.
Many applications include priors for the base measure \citep{escowest_1995,mulleretal_1996} and incorporate estimation of  $Q_0$ and $\alpha$ into the
posterior inference.
Other approaches have addressed the challenge of specifying $Q_0$ by applying empirical Bayes techniques to develop a point estimate for $Q_0$ \citep{mcauliffeetal_2006}.
In many applications, the base measure is given an overdispersed form in an attempt to avoid an unduly informative prior. Of course, doing so precludes
the incorporation of prior information into the inference.

The particular
case of the DP prior illustrates the general challenge of
incorporating prior information in a nonparametric setting.
The results of \citet{yamato_1984} and \citet{lijoi_regazzini_2004} can be extended to adjust $\alpha$ and $Q_0$ in normal DPMMs  so that the induced
prior expectation and variance of the population mean can be
approximately specified
(as will be discussed further in Section 3),
although specification beyond the population mean is problematic.
\citet{bushetal_2010} proposed a limit of Dirichlet process  approach in order to allow calibration of a minimally informative Bayesian analysis with prior information.
A central part of this effort is to compensate for an overdispersed base measure by developing techniques for setting a local mass property. This is designed to make improper base measures feasible and to address the general problem of base measure elicitation in nonparametric analysis.
\citet{moala_ohagan_2010} proposed a method to update a Gaussian process (GP) prior with expert assessments of the mean and other aspects of an unknown density. 
As with the Dirichlet process prior, the GP prior requires specification of the mean and covariance functions that characterize the GP. These provide a base for the prior in the same way that the $Q_0$ base measure does for the Dirichlet process prior.
In the \citeauthor{moala_ohagan_2010} approach, elicitation of these quantities is derived from expert assessments of quantiles of the unknown distributions.

In this paper, we propose a very general method  that allows for the
combination of
an arbitrary prior on a finite set of functionals with
a nonparametric prior on the remaining aspects of the high- or infinite-dimensional unknown parameter.
In the next section we show how such a
partially informative  prior distribution can be constructed from
the combination of any
prior distribution on the functionals of interest with
the conditional distribution of the parameter given the functionals
under a canonical nonparametric prior.
We show that the resulting marginally specified prior (MSP) inherits
desirable features from the canonical prior:
The MSP will generally share the support of the canonical
prior,
and posterior approximation under  the MSP
can typically be made via small modifications to any Markov chain Monte Carlo
algorithm applicable under the canonical prior.

In Section 3 we illustrate
the use of the marginally specified prior in the context of
multivariate density estimation using normal DPMMs.
In an example, we show that efforts to make the canonical DPMM
informative in terms of marginal means and variances can lead to
poor density estimates, whereas a noninformative DPMM can
lead to  suboptimal estimates of functionals due to its inability
to incorporate prior information. In contrast, a
marginally specified
prior is able to both incorporate prior information and provide
accurate density estimation. Additionally, for this particular example,
accurate prior information results in improved density estimation
over a canonical noninformative nonparametric prior.

In Section 4 we examine the important  problem of
NP Bayes analysis of large sparse contingency tables in 
the presence of prior information on the margins. 
In this context, we develop a  marginally specified prior from a canonical NP Bayes approach.
In an example, we illustrate how canonical NP Bayes methods designed to be informative on the margins can result in poor performance 
in terms of margin-free functionals (such as dependence functions). 
In contrast, a marginally specified prior can 
accommodate prior information about  the population margins  while 
being minimally informative about other aspects of the population, 
resulting in strong performance in terms of both marginal and margin-free 
aspects of the population. 
A discussion of the results and directions for future research follows in Section 5.

\section{Marginally specified priors: Construction and computation}
We consider the general problem of Bayesian inference for a parameter $f$
belonging to  a 
high- or infinite-dimensional space $\mathcal F$. 
For example, Section 3 considers 
multivariate density estimation
over the space 
of all densities on $\mathbb{R}^p$ with respect to Lebesgue measure, 
and Section 4 considers the high-dimensional space of 
multiway  contingency tables. 
In general, 
Bayesian inference for $f$ is based 
on a posterior distribution $\pi(f\in A|y)$
derived from a sampling model $\{ p(y|f) : f\in \mathcal F\}$
and a prior distribution $\pi$ defined on a $\sigma$-algebra $\mathcal A$
of $\mathcal F$. 
In many high-dimensional problems 
there are only a few classes of priors for which posterior inference is tractable. Typically, practitioners choose a member $\pi_0$ of such a class 
based on support considerations and the feasibility of posterior 
approximation, rather than 
how well it accurately represents any 
information we may have about 
specific features of $f$. 
In this section, we show how to construct a nonparametric 
prior $\pi_1$ that is informative about specific features of $f$, but 
has the same support as $\pi_0$ and 
is ``close'' to $\pi_0$ in terms of Kullback-Leibler 
divergence. We also show how MCMC  approximation methods for $\pi_0$ 
can be modified to obtain posterior inference under $\pi_1$. 
 
\subsection{Construction of a marginally specified prior}
Let $\theta=\theta(f)$ be a function of $f$, such as a
population mean of $p(y|f)$, variance, 
marginal probability vectors or some finite set of functionals, 
and let $\Theta$ be the range of  $\theta$. 
Any prior distribution $\pi_0$ on $\mathcal F$  induces a  
prior distribution $P_0$ on $\Theta$ defined by 
\begin{align}\label{E:NobodysPrior}
P_0(B) = E_{\pi_0}[ 1 (\theta \in B )  ] , 
\end{align}
where $B$ is any element of $\mathcal{B}$, a $\sigma$-algebra of  ${\Theta}$ 
making $\theta(f)$ a measurable function. 
If $\pi_0$ is chosen for computational convenience, 
the induced prior $P_0$ may not show substantial agreement with available prior information $P_1$ for the functional $\theta$. 
In some cases it may be possible to select a prior $\pi_0$ 
from a computationally feasible class 
to make the induced prior $P_0$ similar to $P_1$:
The results of 
\citet{lijoi_regazzini_2004} and \citet{yamato_1984} provide 
some guidance for Dirichlet process priors 
 if the functionals are means, but in general this will be difficult. 
Furthermore, depending on the structure of the nonparametric class, 
selecting $\pi_0$ in order to match $P_0$ to $P_1$ may result in 
$\pi_0$ being inappropriate for other aspects of $f$. 
As will be illustrated in an 
 example in Section 3, it can be difficult to make 
$\pi_0$ highly informative about $\theta(f)$ but weakly informative 
about other aspects of $f$. 

Suppose a nonparametric prior $\pi_0$ has been 
identified that is viewed as reasonable 
in some respects, such as being computationally feasible and having a
large support, 
but does not represent available prior information 
$P_1$ about $\theta$.
The information in $P_1$ can be accommodated by 
replacing $P_0$, the $\theta$-margin of $\pi_0$, with 
the desired margin $P_1$. 
Specifically, a marginally specified prior (MSP) $\pi_1$  for $f$ 
is obtained by combining the conditional 
distribution of $f$ given $\theta$ with our desired marginal distribution $P_1$ 
for $\theta$, so that 
\begin{eqnarray} \label{E:SomebodysPrior}
\pi_1(A) = \int \pi_0(A|\theta)  P_1(d\theta) .
\label{eq:msp}
\end{eqnarray}
Since $\theta=\theta(f)$, 
 $\pi_0(A|\theta)$ is a random function of 
$f$ and is not uniquely defined 
on null sets of $\pi_0$. To make (\ref{eq:msp}) meaningful, 
we restrict attention to informative 
prior distributions such that 
$P_1$ is dominated 
by $P_0$. 
Under this condition, 
the measure $\pi_1$ on $\mathcal A$ is well defined, 
and the $\theta$-marginal of $\pi_1$ can be computed as  
\begin{eqnarray*}
\pi_1( \{ f: \theta\in B\})  &=&
 \int  \pi_0(\{f: \theta\in B \}|\theta )  P_1(d\theta) \\  
&=& \int 1(\theta\in B) P_1(d\theta)  \\
&=& P_1(B) 
\end{eqnarray*}
for $B\in \mathcal B$
as was desired. 
Additionally, since $P_1 \ll P_0$, these measures have densities 
$p_1$ and $p_0$ with respect to a common dominating measure $\mu$  
(which can be taken equal to $P_0$, for example). 
This allows us to easily relate  the 
support  of $\pi_1$ to that of $\pi_0$:
\newtheorem{lemma}{Lemma}
\begin{lemma}\label{T:PiOneSupport}
Suppose  $P_1 \ll P_0$.  Then 
$\pi_1(A) =  E_{\pi_0}[  1(f\in A)  \tfrac{p_1(\theta)}{p_0(\theta)} ]$
for $A\in \mathcal A$. 
\end{lemma}
\begin{proof} 
Let $B_0 = \{ \theta: p_0(\theta)> 0 \} $. Then 
 $1=P_0(B_0) =P_1(B_0)$ by the assumption and so 
\begin{eqnarray*} \label{E:PiOneProof}
\pi_1(A)&= &  
  \int_{B_0} \pi_0(A|\theta) p_1(\theta) \mu(d\theta) \\
   &=& \int_{B_0} \pi_0(A|\theta)  \tfrac{p_1(\theta) }{p_0(\theta) }  p_0(\theta) \mu(d\theta)  \\
 &=& \int E_{\pi_0} [1(f\in A) \tfrac{p_1(\theta) }{p_0(\theta) }   |\theta]  p_0(\theta) \mu(d\theta) \\
&=& E_{\pi_0} [ 1(f\in A)  \tfrac{p_1(\theta)}{p_0(\theta) }  ]. 
\end{eqnarray*}
\end{proof}

As a corollary, 
if the support of $p_1$ matches that of $p_0$, then 
the support of $\pi_1$ will be that of $\pi_0$: 
\newtheorem{cor}{Corollary}
\begin{cor}\label{T:PiOneSupport2}
Suppose $P_1 \ll P_0 \ll P_1$. Then $\pi_1 \ll \pi_0 \ll \pi_1$. 
\end{cor}
\begin{proof}
It is clear from the definition of $\pi_1$  that $\pi_1\ll\pi_0$. 
To show $\pi_0\ll \pi_1$, 
let $A\in \mathcal A$ be a set such that  $\pi_1(A) = 0$. 
We will show that $P_0\ll P_1$  implies $\pi_0(A)=0$. 
Let $B_j = \{ \theta : p_j(\theta)>0 \}$ and 
$A_j=\{ f: \theta(f)\in B_j\}$ so that 
$\pi_j(A_j) = P_j(B_j) =1$ for $j\in \{0,1\}$. We have 
\begin{eqnarray}
0 = \pi_1(A) &=& \pi_1(A\cap A_1) \nonumber \\
  &=& E_{\pi_0}[ 1(A\cap A_1 ) \tfrac{p_1}{p_0} ] \nonumber \ \ \mbox{(by Lemma 1)} \\
  &=& E_{\pi_0}[ 1(A\cap  A_0  \cap A_1 ) \tfrac{p_1}{p_0} ] . 
\label{eq:supeq1}
\end{eqnarray}
Since $p_1/p_0 >0$ on $A_0\cap A_1$, (\ref{eq:supeq1}) implies that 
$\pi_0(A\cap A_0 \cap A_1)= 0$. Since $\pi_0(A_0)=1$, we have 
$\pi_0(A \cap A_1 ) = \pi_0(A) - \pi_0(A \cap A_1^c)  = 0$. 
Since $0=\pi_1(A_1^c) = P_1(B_1^c)$
and $P_0\ll P_1$, 
we must have 
  $0 =P_0( B_1^c) = \pi_0(A_1^c)$, and 
so $\pi_0(A) =0$. 
\end{proof}

We also note that  $\pi_1$ has a characterization as the 
prior distribution that is closest to $\pi_0$ in terms of Kullback-Leibler divergence, among  priors 
with $\theta$-marginal density equal to $p_1$. The divergence 
of any prior $\pi_1$ dominated by $\pi_0$ is given
by $E_{\pi_0}[\ln \tfrac{ \pi_1(f) }{\pi_0(f)} ]$, where the 
densities can be taken to be with respect to the $\pi_0$-measure, 
and here and in what follows $\pi$  denotes either a measure or a 
density, depending on context. 
If $\pi_1$ has $\theta$-marginal density $p_1$, the divergence can be expressed
as 
\begin{eqnarray*}
E_{\pi_0}[\ln  \tfrac{ \pi_1(f) }{\pi_0(f)} ] &=& 
   E_{\pi_0}[\ln  \tfrac{ \pi_1(f|\theta) }{\pi_0(f|\theta)} ] + 
  E_{\pi_0}[\ln  \tfrac{ p_1(\theta) }{p_0(\theta)} ], 
\end{eqnarray*}
which is minimized by setting $\pi_1(f|\theta) = \pi_0(f|\theta)$.


\subsection{Posterior approximation under MSPs}
For practical reasons the most commonly used priors
are those for which there exist  straightforward
Gibbs samplers or Metropolis-Hastings algorithms for posterior approximation.
In many cases, simple modifications to these 
algorithms can be made to allow for the incorporation of informative priors over functionals of interest.
To illustrate, suppose that under prior $\pi_0$ we have a Gibbs sampler for a high dimensional 
parameter $f$.
Recall that the Gibbs sampler can be viewed as a Metropolis-Hastings 
algorithm for which the proposals are accepted with probability one. 
From this perspective, a Gibbs sampler for  approximating the 
posterior density $\pi_0(f|y)$ is constructed from proposal distributions
with densities 
 $J(f^*| f, y)$ that are proportional to the posterior density, so that
\begin{equation}
 \frac{J(f^*|f,y) }{J(f|f^*,y) } = \frac{ \pi_0(f^*|y) }{\pi_0(f|y) }. 
\label{eq:gibbsprop}
\end{equation}
For example, decomposing $f$ as  $\{f_1,\ldots, f_K\}$, 
the full conditional distribution $\pi_0(f_k|f_{-k},y)$ is one such proposal distribution. 

Posterior approximation of $\pi_1(f|y)$ can proceed by using the 
proposal distributions of the Gibbs sampler for $\pi_0(f|y)$, but adjusting 
the acceptance probability. Specifically, the algorithm for 
approximating $\pi_1(f|y)$ 
proceeds by iteratively simulating proposals $f^*$ from 
distributions of the form 
  $J(f^*|f,y)$ which satisfy (\ref{eq:gibbsprop}), 
and accepting each proposal $f^*$ with probability 
$1\wedge r_{\rm MH}$, where
\begin{eqnarray*}
r_{\rm MH} &=& \frac{\pi_1(f^*|y)}{\pi_1(f|y)} \times \frac{J(f|f^*,y) }{J(f^*|f,y) }  \\
  &=& \frac{\pi_1(f^*|y)}{\pi_1(f|y)} \times 
 \frac{ \pi_0(f|y) }{\pi_0(f^*|y) }  \\
  &=& \frac{ p(y|f^*) \pi_1(f^*) }{p(y|f) \pi_1(f) }  \times 
      \frac{ p(y|f) \pi_0(f) }{p(y|f^*) \pi_0(f^*) }  
 = \frac{  \pi_1(f^*)/\pi_0(f^*)  }{\pi_1(f)/\pi_0(f)  }. 
\end{eqnarray*}
If $\pi_1$ is a marginally specified prior based on $\pi_0$ 
and a marginal density $p_1$ for  $\theta=\theta(f)$, 
we can write 
$\pi_1(f) = \pi_1(\theta) \pi(f|\theta)   = p_1(\theta) \pi_0(f|\theta)$, 
so that 
the acceptance ratio simplifies to
\[   \frac{ p_1(\theta^*)/p_0(\theta^*) }{ p_1(\theta)/p_0(\theta)  }. \]

Similarly, an approximation algorithm for $\pi_1(f|y)$ 
can be constructed from a Metropolis-Hastings algorithm for 
$\pi_0(f|y)$ via the same adjustment. Suppose we have 
a proposal distribution $J(f^* | f,y)$ such that 
the acceptance ratio  $r^0_{\rm MH}$ for $\pi_0$ is computable:
\[
r_{\rm MH}^0 = \frac{ \pi_0(f^*|y) }{\pi_0(f|y)} \frac{ J(f|f^*,y)}{J(f^*|f,y)}
\]
The Metropolis-Hastings algorithm for approximating 
$\pi_1(f|y)$ using $J(f^*|f,y)$ has acceptance ratio
\begin{eqnarray*}
r_{\rm MH}  &=&  \frac{\pi_1(f^*|y) }{\pi_1(f|y)} \frac{ J(f|f^*,y)}{J(f^*|f,y)} \\ 
&=& \frac{\pi_1(f^*|y) }{\pi_1(f|y)} \frac{ \pi_0(f|y) }{\pi_0(f^*|y)} r_{\rm MH}^0  \\
&=& \frac{ p_1(\theta^*)/p_0(\theta^*)} { p_1(\theta)/p_0(\theta)}r_{\rm MH}^0. 
\end{eqnarray*}

These results show that
an MCMC  approximation to $\pi_1(f|y)$ 
can be constructed from an MCMC algorithm for   
$\pi_0(f|y)$ as long as the ratio $p_1(\theta)/p_0(\theta)$ can be computed. 
The value of $p_1(\theta)$ for each $\theta\in \Theta$ is presumably
available as $p_1$ is our desired prior distribution for $\theta$.
In contrast, obtaining a formula for
$p_0(\theta)$ may be difficult.
In situations where  the dimension of $\theta$ is not too large, one
simple solution is to obtain a Monte Carlo 
estimate of $p_0$ based on samples of $f$ from $\pi_0$.
Specifically, we can obtain an i.i.d.\ sample
$\{ \theta_i = \theta(f_i) , i=1,\ldots, S\}$ from  
$f_1,\ldots, f_S\sim$ i.i.d.\ $\pi_0$, and then approximate
$p_0$ with 
a kernel density estimate or flexible parametric family. 
Note that this can be done before the Markov chain is run, so that
the same estimate of $p_0$ is used for each iteration of the
algorithm. 

In situations where obtaining a reliable estimate of $p_0$  is not
feasible, it is still possible 
to  induce  a prior $p_1$ that is approximately equal 
to a  target prior $\tilde p_1$, 
as long as $p_0$ is relatively flat compared to $\tilde p_1$. 
This can be done by replacing $p_0$, the $\theta$-margin of 
$\pi_0$, with
$ p_1(\theta)  \propto p_0(\theta) \tilde p_1(\theta)  
     =  K p_0(\theta) \tilde p_1(\theta)$. 
This defines a valid probability density as long 
as $p_0\tilde p_1$ is integrable, 
which is the case, for example, if either density is bounded. 
Heuristically,
if the prior $\pi_0$ on $\mathcal F$  is chosen to be very diffuse,
then the induced prior $p_0$ is likely to be relatively
flat on $\Theta$ compared to the target informative prior
$\tilde p_1$, and we should have $p_1 \approx \tilde p_1$.
In terms of the MCMC approximation to the resulting marginally specified prior $\pi_1$, 
the adjustment to the acceptance ratio is then 
\begin{eqnarray*}
 \frac{ p_1(\theta^*)/p_0(\theta^*) }{ p_1(\theta)/p_0(\theta)  } &=&
  \frac{ \tilde p_1(\theta^*) }{\tilde p_1(\theta) },
\end{eqnarray*}
which is presumably computable as $\tilde p_1$ is the desired 
prior density.


\section{Density estimation with marginally adjusted DPMM}
\label{S:PostSampMeth}
Perhaps the most commonly used 
NP Bayes procedure is the Dirichlet process mixture model, 
or DPMM \citep{lo_1984,escowest_1995, maceachern_muller_1998}. 
The DPMM consists of  a mixture model along with a Dirichlet process prior 
for the mixing distribution. The population density to be estimated 
and the prior 
can be  expressed as 
\begin{eqnarray*}
p(y|Q) &= &\int p(y|\psi) Q(d\psi)  \\
Q &\sim& {\rm DP}(\alpha Q_0) , 
\end{eqnarray*}
where 
$\alpha$ and $Q_0$ are hyperparameters of the Dirichlet process prior, 
with $Q_0$ typically chosen to be conjugate to the 
parametric family of 
mixture component
densities,
$\{ p(y|\psi) : \psi\in \Psi  \}$, to facilitate posterior 
calculations. 
In this section we show how to
obtain posterior approximations  under
a marginally specified prior $\pi_1$ based on a
DPMM. 
The approach is illustrated with the 
specific case 
of multivariate density estimation, for which 
we take 
the parametric family to 
be the class of multivariate normal densities.
In an example analysis of the well-known 
bivariate dataset on eruption times of the Old Faithful geyser, 
we construct a prior distribution $\pi_1$ based on the multivariate normal 
DPMM with  a marginally specified informative prior on the marginal means and variances. 
Inference under 
$\pi_1$ is compared to 
inference under two standard DPMMs,
one where the  hyperparameters
are chosen to be informative about $\theta$ and another 
where the hyperparameters are noninformative. 

\subsection{Posterior approximation}
Given a sample $\v y_1,\ldots, \v y_n \sim $ i.i.d.\ $p(\v y|Q)$, 
posterior approximation for conjugate DPMMs is often  made with a 
Gibbs sampler that iteratively simulates values of a 
 function that associates data indices to 
the atoms of $Q$. In a DPMM, since $Q$ is discrete 
with probability one, a given  mixture component 
(atom of $Q$) may be associated with multiple observations. 
Let $g:\{ 1,\ldots, n\} \rightarrow \{1,\ldots, n\}$ be the 
unknown 
mixture component membership function, so that
$g_i=g_j $ means that $\v y_i$ and $\v y_j$ came from the 
same mixture component. Note that $g$ can always be expressed as 
a function that maps  $\{1,\ldots,n\}$ onto $\{1,\ldots, K\}$, where 
$K\leq n$. 
Inference for conjugate DPMMs often
proceeds by iteratively sampling each $g_i$ from 
its full conditional distribution  $p(g_i| \v y_1,\ldots, \v y_n , g_{-i} )$ \citep{bush_maceachern_1996}.
Additional features of $Q$ and $p(\v y|Q)$ can be simulated given 
$g_1,\ldots, g_n$ and the data. 


This standard algorithm for DPMMs can 
be modified to accommodate 
a marginally specified 
prior distribution on a parameter $\theta=\theta(Q)$. 
Let $f = \{ g,\theta \}$ and let $\pi_0$ be the prior density on $f$ 
induced by the Dirichlet process on $Q$. 
Our marginally specified prior is given by 
 $\pi_1(f) = \pi_0(f) p_1(\theta)/p_0(\theta)$, 
where $p_0$ is the density for  $\theta$ induced by $\pi_0$
and $p_1$ is the informative prior density.
An MCMC approximation to $\pi_1(f|y_1,\ldots,y_n)$
can be obtained via 
 the procedure outlined in Section 2.2. 
Given a current state of the Markov chain $f=\{\theta ,
 g_k,g_1,\ldots,g_{k-1},g_{k+1},\ldots, g_n\} = \{\theta,g_k,g_{-k}\} $, 
the next state is determined as follows:
\begin{enumerate} 
\item 
Generate a 
proposal 
$f^*=\{\theta^*,g_k^*,g_{-k}\}$  
from 
$\pi_0(\theta,g_k| g_{-k},y) = \pi_0(g_{k}|g_{-k},y) \pi_0(\theta| g,y )$ by
\begin{enumerate}
\item generating $g_k^* \sim \pi_0(g_k | g_{-k},y) $; 
\item generating $\theta^* \sim \pi_0(\theta| g_k^*,g_{-k},y)$.
\end{enumerate}
\item
Set the value of the next state  of the chain to  $f^*$ with probability 
 $1\wedge [p_1(\theta^*)/p_0(\theta^*) ] /
   [p_1(\theta)/p_0(\theta)] $, otherwise let the next state equal 
the current state. 
\end{enumerate}
This procedure is iterated over values of $k\in \{1,\ldots, n\}$, 
possibly in random order, 
and repeated until the desired number of simulations of $f$ is obtained. 
Note that steps 1.(a) and 1.(b) compose a standard Gibbs sampler 
for the DPMM in which posterior inference for $\theta$ is provided, 
although typically we would only simulate $\theta$ once
per complete update of $g_1,\ldots, g_n$. 
The algorithm for the  marginally specified prior 
 $\pi_1$ requires that 
$\theta$ be simulated with each proposed value of $g_k$ 
so that the acceptance probability in 
step 2 can be calculated. 

Implementing the steps of this MCMC algorithm involves two non-trivial computations: simulation of $\theta$ from $\pi_0(\theta|g,y)$, and calculation of 
$p_0(\theta)$ in order to obtain the acceptance probability. 
General methods for the latter were discussed in Section 2.2. 
For the former, 
we suggest using a 
Monte Carlo approximation to $Q$ 
based upon a representation of Dirichlet processes due to \citet{pitman_1996}. 
Let $K$ be the number of unique values of $g_1,\ldots, g_n$ and 
let $n_k$ be the number of observations $i$ for which $g_i=k$.
If $Q_0$ is conjugate, then the parameter values 
$\psi_{(1)},\ldots, \psi_{(K)}$  corresponding to the 
mixture components can generally be easily simulated. 
Corollary 20 of \citet{pitman_1996} gives 
the conditional distribution of $Q$ given 
$\psi_{(1)},\ldots, \psi_{(K)}$  and counts $n_1,\ldots, n_K$ as 
\[ \{ Q(H) |\psi_{(1)},\ldots, \psi_{(K)}, n_1,\ldots, n_K  \} \stackrel{d}{=} \gamma \sum_{k=1}^K 1(\psi_{(k)}\in H)  w_k +
(1-\gamma) \tilde Q(H), \]
where $\gamma \sim {\rm Beta}(n,\alpha)$,  $w \sim {\rm Dirichlet}(n_1,\ldots, n_K)$ and  $\tilde Q\sim {\rm DP}(\alpha Q_0)$.
A Monte Carlo approximation to $Q$, 
and therefore any functional of $Q$,  can be obtained via simulation of a
large number $S$  of $\psi$-values from $Q$.
To do this, we first simulate $\gamma$ and $w_1,\ldots, w_K$
from their beta and Dirichlet full conditional distributions.
From these values we sample cluster memberships for a sample of size
$S$ from $Q$ using a multinomial$(S, \{ \gamma w_1,\ldots, \gamma w_K,1-\gamma \})$ distribution. Note that the count $s$ for the $K+1$st category
represents the number of $\psi$-values that must be simulated from
$\tilde Q$.
To obtain
the sample from $\tilde Q$ we run a Chinese restaurant process
of length $s$, and then generate the unique $\psi$-values from
$Q_0$ for each partition. This can generally be done quickly
for two reasons:
First,
the expected number of samples needed from $\tilde Q$ is only $S \alpha/(\alpha+n)$.
For example, with $S=1000$, $n=30$ and $\alpha=1$, we expect to only need about $s=32$
simulations from $\tilde Q$.
Second, the number of unique values in a sample of size $s$ from $\tilde Q$ is only of order $\log s$, which will generally be manageably small.

\subsection{Example: Old Faithful eruption times}
The Old Faithful dataset consists of 272 bivariate observations
of eruption times and waiting times between eruptions, both measured in 
minutes. 
To illustrate and evaluate 
the MSP  methodology
we construct two subsets of these data: a random sample  
of size $n_0=30$ from which we obtain prior information and 
a second, non-overlapping random sample of size $n=30$ representing 
our observed data. 
The random samples were obtained by setting
the random seed in {\sf R} (version 2.14.0) to 1, sampling the prior
dataset, and then sampling
the observed dataset from the remaining observations.
For the purpose of this example, we view 
the full dataset of 272 observations as  the 
``true population.'' 
A scatterplot of the observed data and  marginal density estimates 
are shown graphically in Figure \ref{fig:data}. 
The observed dataset consisting of $n=30$ observations 
clearly captures the bimodality of the population. However, the 
marginal plots indicate that the sample has overrepresented 
one of the modes. 

\begin{figure}[ht]
\centerline{\includegraphics[width=6.75in]{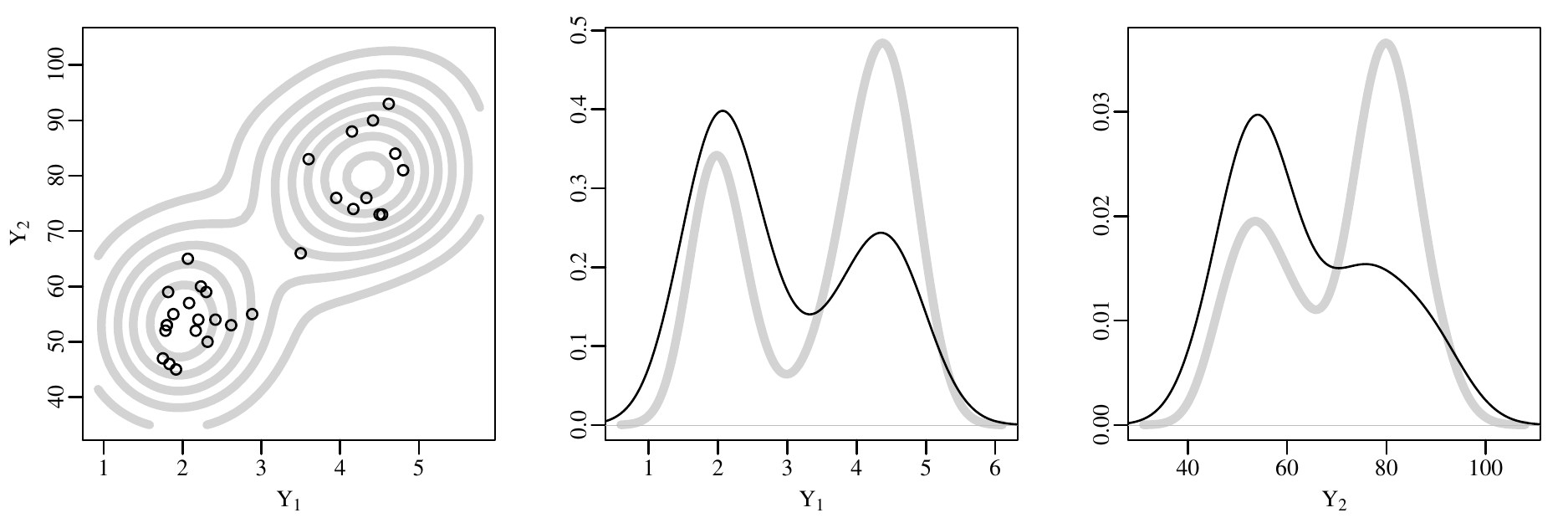}}
\caption{Population and sample: The left-most panel 
shows the contours of the population density and a scatterplot of 
the $n=30$  randomly sampled observations. 
The center and right panels show marginal densities for the 
population (light gray) and sample (black). }
\label{fig:data}
\end{figure}

Suppose our knowledge of the prior sample is limited to 
the bivariate marginal  sample means $m_0 \in \mathbb R^2$ and 
sample variances 
 $v_0 \in \mathbb (R^+)^2 $. 
In such a situation it would be desirable to construct 
a prior density $p_1$ over the unknown population marginal means $m$
and variances $v$ based on the values of  $m_0$, $v_0$ and $n_0$, and 
combine this information with the 
information in our fully observed sample to improve our inference about the population. 
Incorporating this information with conjugate priors 
would be more or less straightforward if our sampling model 
were bivariate normal, but it is 
difficult in the context of a DPMM. 
Proposition 5 of 
\citet{yamato_1984} indicates that if the base measure 
$Q_0$  in the Dirichlet process prior is multivariate 
normal$(\mu_0,\Sigma_0)$, 
then 
the induced prior distribution on the mean $\int x Q(dx)$ is 
approximately multivariate normal$(\mu_0,\Sigma_0/[\alpha+1]) $. 
This result is not directly applicable to the multivariate 
normal DPMM  for two reasons, one being that $Q$ represent the 
mixing distribution and not the population distribution, 
and the other being that
in the conjugate multivariate normal DPMM the 
parameter $\psi$ in the mixture component consists not just of 
a mean $\mu$ but also a covariance matrix 
$\Sigma$. 
Specifically, in the conjugate $p$-variate normal DPMM,  the density $q_0$ of the base  measure $Q_0$ for  $\psi=(\mu,\Sig)$ is given by 
\begin{equation}
q_0(\v\mu, \Sig)  = {\rm normal}_p(\v\mu : \v\mu_0,\Sig/\kappa_0) \times 
   \mbox{inverse-Wishart} (\Sig: \m S_0^{-1}, \nu_0) 
\label{eq:dpmbm}
\end{equation}
where the functions on the right-hand side  are the multivariate 
normal and inverse-Wishart densities respectively, 
the latter being parameterized  so that 
$\Exp{\Sig} = \m S_0/(\nu_0-p-1)$. 
With some effort (details available from the second author) it is possible to 
obtain values of 
the hyperparameters 
$(\mu_0,\kappa_0,S_0,\nu_0)$ and $\alpha$ so that the 
induced prior  distributions on the 
population mean $m(Q) =\int \int y  p(y|\psi) Q(d\psi) dy $ 
and variance $V(Q) = \int \int y y^T p(y|\psi) Q(d\psi) dy - m(Q)m(Q)^T$ 
have the following properties: 
\begin{equation}
\Exp{ \v m (Q) }  \approx  \v m_0   \ , \ \ \ 
\Var{ \v m (Q) }  \approx   \m V_0/n_0   \ , \ \ 
\Exp{ \m V(Q) } \approx  \m V_0 .  
\label{eq:infmom}
\end{equation}
Unfortunately, it seems difficult to specify the prior  on
$\m V(Q)$  separately from that of $\m m(Q)$ within the context of the DPMM. 

We construct three different 
nonparametric prior distributions for 
a comparative 
analysis of the Old Faithful data:
\begin{itemize}
\item Informative DPMM $\pi_0^I$: 
The base measure  density
$q_0$ is as in (\ref{eq:dpmbm}) with 
$(\mu_0=m_0,\kappa_0=n_0/(\alpha+1) , \nu_0=n_0, S_0 = \nu_0 V_0)$, 
where the diagonal of $V_0$ is $v_0$, the marginal variances from the 
prior sample, and the correlation is equal to the sample correlation from the 
observed data. 
This results in a prior on $Q$ essentially satisfying
(\ref{eq:infmom}), 
thereby utilizing the prior information.

\item Noninformative DPMM $\pi_0^N$:   The base measure density 
$q_0$ is as in (\ref{eq:dpmbm}) with
$(\mu_0=\bar y,\kappa_0=1/10 , \nu_0=p+2=4, S_0 = S_y )$, 
where $\bar y$ is the sample mean from the $n=30$ 
values in the observed sample, and $S_y$ is the sample 
covariance matrix. This prior does not use information from the prior sample, and is designed to promote relative  diffuseness of the induced prior 
on the marginal population  means and variances.
Note that using sample moments for the hyperparameters
weakly centers the prior around the observed data. 
We can view this as a type of ``unit information'' prior 
\citep{kass_wasserman_1995}. 

\item Marginally specified prior $\pi_1$: Letting $\theta=(m_1,m_2,v_1,v_2)$
be the unknown population means and marginal variances, 
we construct a marginally specified prior by replacing the 
$\theta$-margin  of $\pi^N_0$ with $p_1(\theta)$, a 
product of two  univariate normal and two inverse-gamma densities, 
chosen to match
the prior on $\theta$ induced by $\pi_0^I$ as closely as possible.
\end{itemize}
Thus $\pi_0^I$ and $\pi_1$ have roughly the same $\theta$-margin, 
but otherwise $\pi_1$ matches the more diffuse prior $\pi_0^N$. 
Of course, we could have given $\pi_1$ any $\theta$-margin we wished, 
but matching the margins of $\pi_0^I$  and $\pi_1$ facilitates comparison. 
The hyperparameter $\alpha$ was set to 1 for all of the above prior 
distributions.

In order to evaluate the Metropolis-Hastings ratios when approximating
the posterior  distribution under $\pi_1$,  we found that a
skewed multivariate $t$-distribution 
provided a very accurate 
approximation to  the joint distribution 
of the marginal means and log variances induced by $\pi^N_0$. 
Via a change of variables, this provides an accurate approximation 
to $p_0(\theta)$, with which the acceptance probability 
is computed for approximation of $\pi_1(f|y)$. 

Markov chains of length 25,000 were run under each prior,
with parameter values being saved every 10th iteration, resulting 
in 2500  simulated values of each parameter with which to make posterior 
approximations. 
The chains showed no evidence of non-stationarity and 
mixed  well under each prior: 
Based on the dependent MCMC sequences of length 2500, 
the equivalent number of 
independent observations of $\theta$ 
(i.e., the effective sample sizes)
were estimated as  above 2000 
for each element of $\theta$ and under each prior. 

Posterior predictive distributions under the 
three priors are shown in Figure \ref{fig:yjoint}. 
The informative DPMM provides a  poor 
representation of  the population distribution, given in light gray contours. 
This is primarily a result of having to set the  $\kappa_0$ 
hyperparameter to be 
moderately large $(\kappa_0=30)$ in order to obtain 
the desired  informative 
prior variance for the population mean $\v m=(m_1,m_2)$. 
Unfortunately, 
setting this parameter so high means that values of $\v \mu$ in the mixture 
model are tightly concentrated around $\v m_0$, and so the multimodality 
is not captured.  
In contrast, the posteriors under the 
 noninformative DPMM $\pi_0^I$ and the MSP $\pi_1$  are able to capture the multimodality of the population. 

\begin{figure}[ht]
\centerline{\includegraphics[width=6.65in]{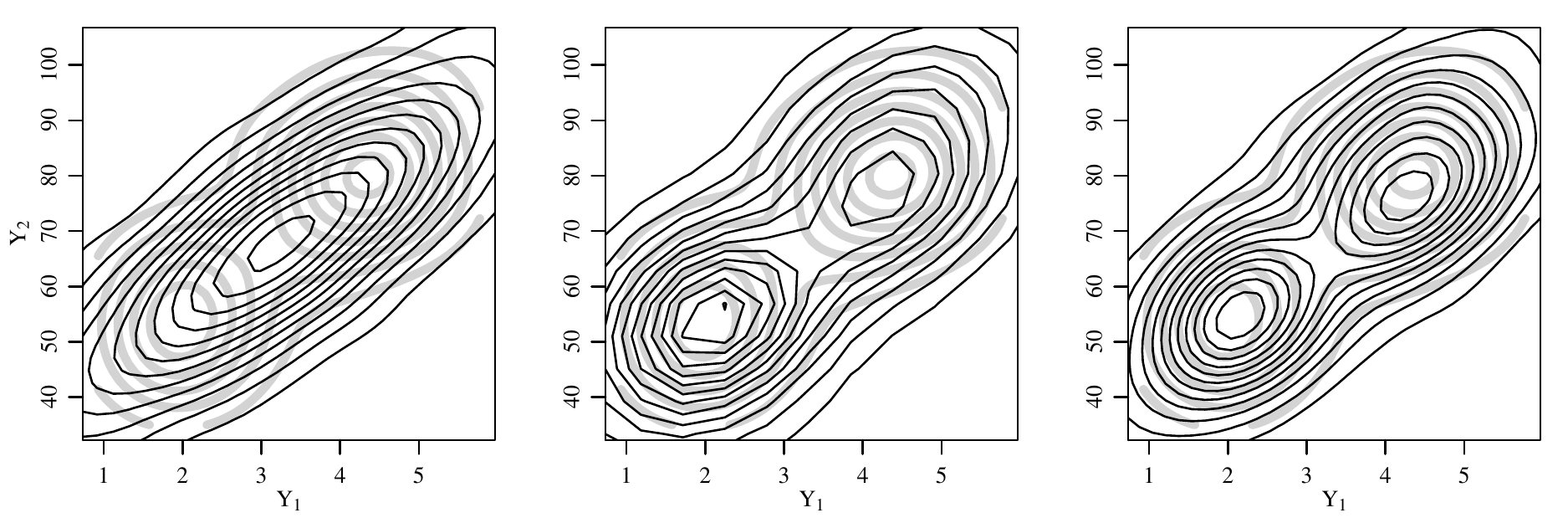}}
\caption{ Contour plots of the posterior predictive density in black
and the  population  density in gray, under 
$\pi_0^I$, $\pi_0^N$ and $\pi_1$ from left to right. 
} 
\label{fig:yjoint}
\end{figure}

Figure \ref{fig:ymarg}
gives marginal  density estimates
under the different priors.
The figure suggests that the posterior under $\pi_1$ is  better at representing the underlying population than the posteriors under the other priors. 
Recall that the observed sample contains an unrepresentative number of 
low-valued observations. 
The posterior under the non-informative prior  $\pi_0^N$
uses 
only the observed data and thus is equally unrepresentative  of the population.
In contrast, $\pi_1$ is able to use some information from the 
prior sample, and is therefore more representative of the population. 

Finally, the marginal posterior distributions of the 
marginal parameters $\v m$ and $\log \v v$ 
are given in Figure \ref{fig:psimarg}. The priors are given 
in gray and the resulting posterior distributions are given in black. 
The population values based upon the full set of 272 observations  are given by gray vertical lines. 
Across all parameters, $\pi_1$  gives posteriors that are most concentrated
around the population means. Note that the difference between the priors
and the posteriors under $\pi_0^I$ is not that large. 
We conjecture that this is primarily a result of the fact that under 
$\pi_0^I$, most observations are estimated as coming from the same 
mixture component, thereby overestimating the entropy, 
when in fact the data are bimodal. 
In contrast, $\pi_1$ is able to recognize the bimodality and obtain 
improved estimates of the marginal densities.  

In this example, we have shown that efforts to make the canonical DPMM
informative in terms of marginal means and variances can lead to
poor density estimates, 
 whereas a noninformative DPMM can
lead to  suboptimal estimates of functionals due to its inability
to incorporate prior information. In contrast, a
marginally specified
prior is able to both incorporate prior information and provide
accurate density estimation.


\begin{figure}[ht]
\centerline{\includegraphics[width=5.5in]{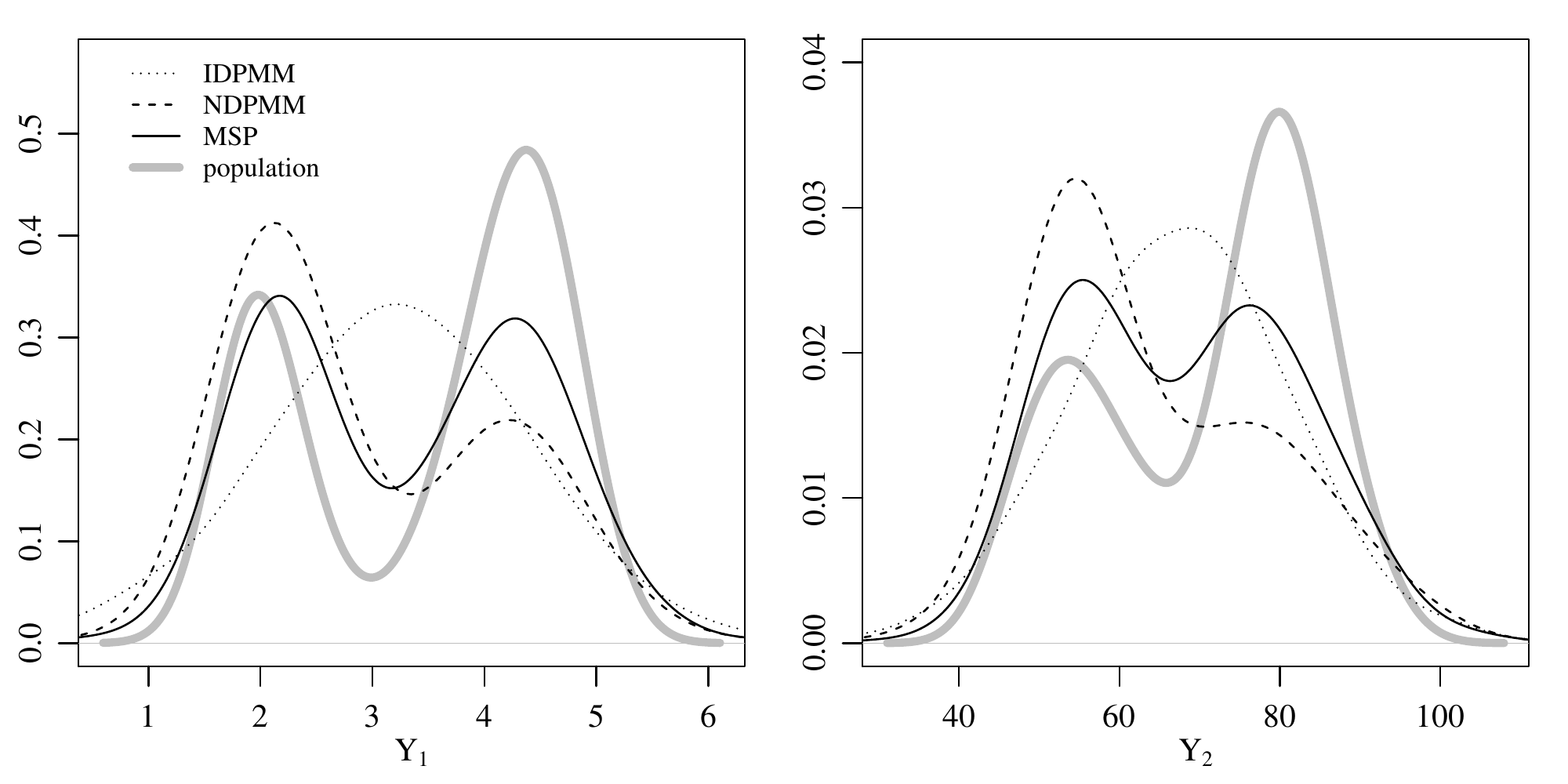}}
\caption{ Marginal population densities and 
estimates 
from the three priors: informative DPMM (IDPMM), noninformative DPMM (NDPMM) and marginally specified prior (MSP).   }
\label{fig:ymarg}
\end{figure}

\begin{figure}[ht]
\centerline{\includegraphics[width=6.65in]{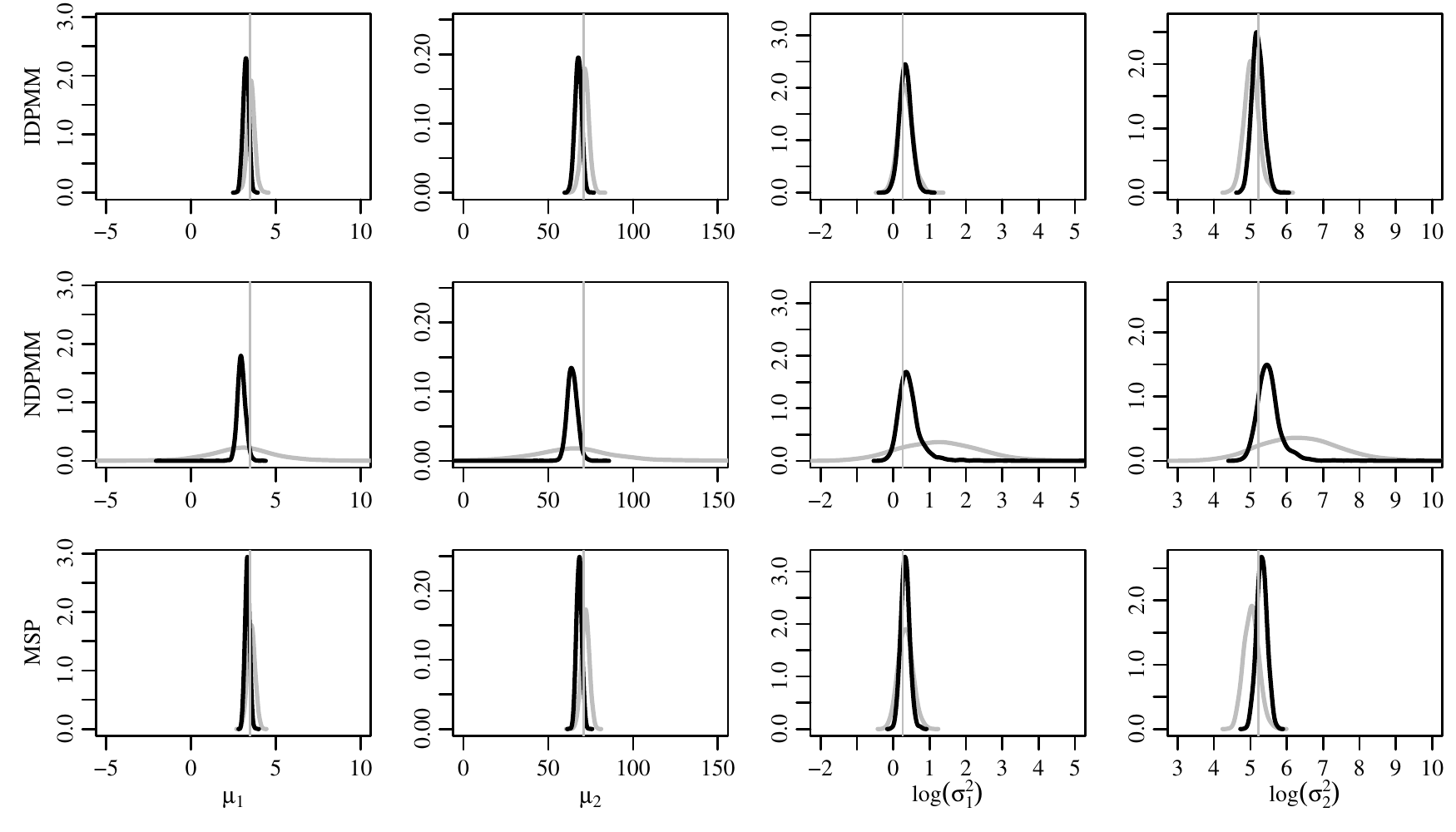}}
\caption{Priors (gray) and posteriors (black) for the marginal 
means  and log variances.}
\label{fig:psimarg}
\end{figure}

\newcommand{\margpri}{\v \theta}
\newcommand{\margpris}{\theta}
\newcommand{\mprinfo}{\v m}
\newcommand{\mprinfos}{m}

\newcommand{\PriorOne}{$\pi_0^I$}
\newcommand{\PriorTwo}{$\pi_0^N$}
\newcommand{\PriorThree}{$\pi_1$}
\newcommand{\PostThree}{$\pi_1(\v f|\v y)$}
\newcommand{\PriorFour}{$\pi_0$}
\newcommand{\PostFour}{$\pi_0(\v f | \v y)$}

\section{Marginally specified priors for contingency table data}
\label{S:CTM}
\newcommand{\PUMSPEE}{p}
\newcommand{\PUMSPEEMONE}{p-1}
\newcommand{\EQPUMSSP}{}
Even when multivariate categorical data include only  moderate numbers of variables and categories, 
large or full models that allow for complex or  
arbitrary multivariate dependence 
can involve a very large number of parameters. 
For example, a full model for the $2\times 3\times 2\times 8\times 12$-way contingency table data we consider later in this section 
requires a 1151-dimensional parameter.
One Bayesian 
approach to the analysis of such data is via model selection among 
reduced log-linear models \citep{dawid_lauritzen_1993,dobra_massam_2010}. 
However, model selection can be difficult even for moderate numbers of 
variables  and categories, 
due to the large number of models with low posterior probability and the resulting difficultly in completely exploring the model space.
An alternative NP Bayes approach 
is provided by \citet{dunson_xing_2009}, who 
developed a prior based on a Dirichlet process mixture of product multinomial distributions. Such a prior
has full support on the parameter space but concentrates prior 
mass near simple submodels. 
However,  this approach lacks a  straightforward method for the incorporation of the type of marginal prior information that is frequently available for categorical data.

In this section we consider an alternative 
NP Bayes approach based on a marginal adjustment to a 
standard Dirichlet prior distribution. 
This approach is relatively straightforward computationally, 
and also allows for the incorporation of prior information 
on specific functionals of the unknown population distribution, 
such as the  univariate marginals.

\subsection{The canonical Dirichlet prior}

Multivariate categorical data consist of 
observations $y_i= (y_{i1}, \ldots, y_{ip})$, for which $y_{ij} \in \{1,2,\ldots,d_j\}$ for $j=1,\ldots,p$.
A $p-$way contingency table is a common representation for such data, in which each cell of the table indicates the count of observations $y_i$ such that $y_{i1}=c_1, \, \ldots, y_{ip}=c_p$ for a specific  
response vector $c=(c_1,\ldots,c_p)$.
The sampling model for a contingency table can be expressed 
as a multinomial distribution, 
where
for each cell $c\in \mathcal C = \{ c: 1\leq c_j\leq d_j, j=1,\ldots, p\}$
we define $f_{c} \equiv \text{Pr}(y_{i1}=c_1,\ldots,y_{ip}=c_p)$. 
The full model of all distributions for  the data
can therefore be indexed by the parameter 
$f=\{ f_{c} :  c\in \mathcal C\}$, which lies in the 
$(\prod d_j -1)$-dimensional simplex. 
Given $n$ i.i.d.\ observations, the likelihood is 
 $L(f|y_1,\ldots, y_n) = \prod_{c_1=1}^{d_1} \times \cdots \times \prod_{c_p=1}^{d_p} f^{\sum 1(y_i=c)   }_{c}$, 
for which a 
standard conjugate prior is the 
Dirichlet distribution with hyperparameter $\v \alpha \in (\mathbb R^+)^ 
   {\prod d_j}$. 
This is a nonparametric prior in the sense that it gives full support on the space of possible values of $\v f$.  

The Dirichlet prior is an appealing choice computationally because of its conjugacy, but this convenience can result in undesirable side effects.  
In particular, choosing what appears to be an uninformative Dirichlet prior for $\v f$ can induce substantial informativeness about the 
marginals $\{ \margpri_1,\ldots, \margpri_p\}$, where $\margpri_j = \{ \theta_{j1},\ldots, \theta_{jd_j}\} =  \{\Pr(y_{ij}=1|f),  \ldots, \Pr(y_{ij}=d_j|f)\}$.
For example, setting $\alpha_c=1$  for each cell $c\in \mathcal C$ results in 
a uniform prior distribution for $f$, often used as a default prior distribution
in the absence of prior information. 
However, the induced prior on the marginals $\margpri_1,\ldots, \margpri_p$ 
is highly informative:
The marginalization properties of the Dirichlet distribution result in $\margpri_j~\sim~$Dir$( \prod_{k\neq j} d_k, \ldots, \prod_{k\neq j}d_k)$, 
which is generally highly  concentrated around the uniform distribution 
on $\{1,\ldots, d_j\}$. 
On the other hand, it is reasonably straightforward to choose values 
of $\alpha_c$ to induce particular marginal Dirichlet priors on 
the $\theta_j$'s, although each marginal prior must have the same 
concentration. However, 
this approach to constructing an informative prior for the margins 
necessarily induces a prior over the remaining aspects of $f$, such as 
the dependence structure, that could be undesirably informative. 

\subsection{A marginally specified prior}
\label{S:CTSAMP}
To overcome these undesirable features of the Dirichlet prior,
we construct a nonparametric prior on $f$ based upon a Dirichlet 
distribution with a low total concentration, but with the induced marginal 
priors for $\theta_1,\ldots, \theta_p$ replaced with 
informative priors to reflect known information. 
Specifically, our prior for $f$ takes the form 
\begin{eqnarray*} \pi_1(f,\theta)  &=& 
    \pi_0(f|\theta) \times  p_1(\theta) \\
 &=& \pi_0(f|\theta) \times \prod_{j=1}^p p_{1j}(\theta_j), 
\end{eqnarray*}
where $\pi_0(f)$ is a Dirichlet$(\alpha_0,\ldots, \alpha_0)$ distribution on the $(\prod d_j-1)$-dimensional simplex and $p_{1j}$ is an informative  Dirichlet distribution on $(d_j-1)$-dimensional simplex. Recall from Section 2 that the marginally specified prior $\pi_1$  
is the closest distribution in Kullback-Leibler divergence to $\pi_0$ that 
has the desired priors on $\theta_1,\ldots, \theta_p$. 
Also note that the methodology does not require that these induced priors 
be Dirichlet, although making them so will facilitate comparison to 
an informative Dirichlet prior distribution on $f$ in the 
example data analysis that follows. 

Estimation of $f$ via the posterior distribution $\pi_1(f|y)$ can 
proceed via an 
MCMC algorithm. 
As in the previous section, we modify an MCMC algorithm 
for simulating from {\PostFour}, the posterior under the canonical nonparametric prior, in order to obtain simulations from  {\PostThree}, the posterior under the marginally specified prior. Our particular MCMC scheme relies on the 
representation of a Dirichlet-distributed random variable as a set of independent gamma variables scaled to sum to one.  
That is, if $Z_c \sim \text{gamma}(\alpha_c,1)$ and $f_c = Z_c/\sum Z_{c'}$, then 
$f \sim \text{Dirichlet}(\alpha_1, \ldots, \alpha_{|\mathcal C|})$. 
We employ an MCMC algorithm that is based upon simulating proposed values of 
$\{\ln Z_c : c\in \mathcal C \}$ from a normal distribution centered at the 
current values. 
Because of the generally high dimension of the parameter $f$,
proposing changes to every element of $f$ simultaneously can 
result in low acceptance rates. 
To avoid this problem, at each iteration of the algorithm we propose changes
to randomly chosen subvectors of $f$.
The steps in a single iteration of the  MCMC algorithm are then as follows:
\begin{enumerate}
\item Generate a proposal $\{f^*, \theta_1^*,\ldots,\theta_p^*\}$:
\begin{enumerate}
\item randomly sample a set of cells $\mathcal C' \subset \mathcal C$; 
\item simulate proposals
     $\{ \log Z_c^* : c \in \mathcal C' \} = 
      \{ {\log Z_c} : c \in \mathcal C' \} + \epsilon $, $\epsilon \sim 
$ normal$(0,\delta I)$;
\item compute the corresponding $f^*$ and 
       marginal probabilities $\theta_1^*,\ldots, \theta_p^*$. 
\end{enumerate}
\item Compute the acceptance ratio $r=r_0r_1$ from 
$r_{0}$, the acceptance ratio for $f$ under $\pi_0$, and 
$r_1$, the marginal prior ratio:
\[ r_0 = \frac{p(y|\v f^*) \pi_0(\v Z^*)}     {p(y|\v f  ) \pi_0(\v Z  )}
  \prod_c (Z^*_c/Z_c )   \   \  ,   \ \ 
r_1 =  \frac{ p_1(\theta^*) /p_0(\theta^*) } { p_1(\theta) /p_0(\theta) }. \]
%
\item Accept $f^*, \theta_1^*, \ldots, \theta_p^*$ with probability 
$1 \wedge r$. 
\end{enumerate}
Note that the ratio $r_0$ includes the Jacobian of the transformation 
from $Z$ to $\ln Z$, as the proposal distribution 
is symmetric on the log-scale. 
The number of cells $|\mathcal C'|$ to update at each step and the 
variance parameter $\delta$ in the proposal distribution
can be adjusted to achieve target acceptance rates. 

As mentioned above, we take $p_1$ to be a product of 
Dirichlet densities representing prior information about the margins 
$\theta_1,\ldots, \theta_p$. To calculate $r_1$ we must also compute 
the corresponding joint distribution $p_0$ of $\theta_1,\ldots, \theta_p$ 
under the Dirichlet distribution $\pi_0$ on $f$. 
We approximate $p_0$ by the product of the prior  marginal
densities of $\theta_1,\ldots, \theta_p$  under $\pi_0$, each of  which are 
 Dirichlet. 
However, we 
note that the $\theta_j$'s are  only approximately 
independent of each other under $\pi_0$.

\subsection{Example: North Carolina PUMS data}
We evaluate the performance                                            
of the marginally specified prior 
and several associated priors in terms of their 
performance under the scenario of a researcher with accurate prior 
information about the marginal distributions of the $p$ categorical variables.
Our scenario is based on 
data from the Public Use Microdata Sample (PUMS) of the American Community Survey, a yearly demographic and economic survey.  We consider data 
on gender (male, female: $d_1=2$), citizenship (native, naturalized, non-citizen: $d_2=3$), primary language spoken (English, other: $d_3=2$), class of worker ($d_4=8$), and mode of transportation to work ($d_5=12$)
from 40,769 survey participants. 
The latter two  variables are each dominated by a single category, 
``employee of private company'' (63.75\%) for worker class  and 
``car, truck or van'' (91.97\%) for transportation. 
 These classifications yield a five-way contingency table with $|\mathcal C|=1,152$ cells. 
\label{S:CTSIMSTUDY}
From these data 
we constructed a ``true'' joint distribution $\tilde{ \v f}$ 
and marginal frequencies $\tilde \theta$ 
by filling out the multiway contingency table with the PUMS data, replacing zero counts in the contingency table with small fractional counts, and normalizing the resulting counts to produce a probability 
distribution over $|\mathcal C|$.
We then simulated smaller datasets of various sample sizes from $\tilde f$, 
 and obtained posterior 
estimates for each under three different prior 
distributions:
\begin{itemize}
\item Informative Dirichlet prior $\pi_0^I$:
A Dirichlet distribution with parameter $\alpha_I f_0^I$, where 
$\alpha_I =|\mathcal C| $ and $f_0^I$ is in the $(|\mathcal C|-1)$-simplex.
Using the method of \citet{csiszar_1975},
the
prior mean $f_0^I$ of $f$ was chosen to be the frequency vector 
closest in  Kullback-Leibler  divergence to the uniform distribution on $|\mathcal C|$
among  those with margins equal to 
$\tilde \theta$. 
The induced marginal prior on  each  $\margpri_j$ is then  $\text{Dir}(|\mathcal C|\tilde \theta_j)$, which has prior  expectation $\tilde \theta_j$ 
as desired. Note that the concentration hyperparameter $\alpha_I$ is the 
same as that for a uniform prior on the simplex. 
\item Noninformative Dirichlet prior $\pi_0^N$:
A Dirichlet distribution with parameter $\alpha_N  f_0^N$, 
where $\alpha_N= \sqrt{ |\mathcal C|}$ and $f_0^N = \{ 1/|\mathcal C|,\ldots, 
1/|\mathcal C|\}$. This prior has the same prior expectation 
as the uniform prior on the $(|\mathcal C|-1)$-simplex, 
but a smaller prior concentration by a factor of $\sqrt{|\mathcal C|}$. 
\item Marginally specified prior $\pi_1$: 
Constructed by replacing the marginal prior for $\theta$ induced by
 $\pi_0^N$ with the marginal prior 
under  $\pi_0^I$. 
\end{itemize}
%

We used the true joint distribution $\tilde{\v f}$ to generate $200$ replicate data sets of sizes $n\in \{$ 100, 1000, 5000, 10000, 20000, 40000 $\}$. 
The $\pi_0^I$ and  $\pi_0^N$ priors are conjugate to the multinomial likelihood, and so their posterior distributions are available in closed form.  
For estimation under $\pi_1$, the MCMC  algorithm described above was run for $3 \times 10^6$ iterations for each simulated dataset. 
The acceptance rate varied with the sample size $n$, from $89\%$ at $n=100$ down to $63\%$ at $n=10000$.
Effective sample sizes corresponding to thinned Markov chains 
based on every 500th iterate were obtained and were found 
to be around 1000 (based on thinned chains of length 6000). 

For each simulated dataset and prior we obtain posterior mean estimates
$( \hat f,\hat \theta)$ which we compare to the true values $(\tilde f, \tilde \theta)$ used to generate 
the simulated data. 
To evaluate $\hat \theta$, 
we use an average of the absolute value of the Kullback-Leibler divergence between the true marginal distributions $\{\tilde \theta_{1}, \ldots \tilde \theta_p \}$ and the estimated marginal distributions $\{ \hat \theta_{1}, \ldots \hat \theta_{p} \}$:
\begin{align*}
M = \frac{1}{\PUMSPEE} \sum_{j=1}^{\PUMSPEE} \bigg | \sum_{c=1}^{d_j} \tilde\theta_{jc} \; \ln \big(\hat\theta_{jc}/\tilde \theta_{jc} \big) \bigg |. 
\end{align*}
Smaller values of $M$ indicate better performance with respect to this  marginal metric.

To assess  the performance of $\hat f$ on aspects of $f$ other than the 
marginal distributions, we compared the true and estimated values of the local dependence functions 
(LDFs) of the $p \choose 2$ separate two-way marginal distributions. 
These LDFs describe the two-way dependencies among the variables, and 
are invariant to changes in the marginal distributions \citep{goodman_1969}. 
The LDFs are formed from cross-product ratios of $f$ as follows:
Letting $f^{j_1,j_2}_{c_1,c_2}=  \Pr( y_{j_1} = c_1, y_{j_2} = c_2|f)$, 
we define 
%
\begin{align*}
LDF^{j_1,j_2}_{c_1,c_2}(f)  = \ln \bigg( \frac{f^{j_1,j_2}_{c_1,c_2} \:\: f^{j_1,j_2}_{c_1+1,c_2+1}}
 {f^{j_1,j_2}_{c_1,c_2+1} \: f^{j_1,j_2}_{c_1+1,c_2}} \bigg). 
\end{align*}
For each simulated dataset and prior distribution, we 
computed the average squared error between 
$LDF^{j_1,j_2}_{c_1,c_2}(\hat f)$ and $LDF^{j_1,j_2}_{c_1,c_2}(\tilde f)$ 
as 
\begin{align*}
L = {p \choose 2}^{-1} \sum_{j_1 < j_2} \frac{1}{(d_{j_1}-1)(d_{j_2}-1)} \sum_{c_1=1}^{d_{j_1}-1} \sum_{c_2=1}^{d_{j_2}-1} (LDF^{j_1,j_2}_{c_1,c_2}(\hat f) - LDF^{j_1,j_2}_{c_1,c_2}(\tilde f))^2. 
\end{align*}
Smaller values of $L$ indicate better performance 
in terms of representing the two-way dependence structure of the 
true distribution $\tilde f$.

Figure \ref{Fi:CTSimStudy} shows the $M$ and $L$ performance metrics for each prior and simulated dataset, with the averages over simulations  at each sample size joined by lines.
The sample sizes are displayed ordinally, with a slight horizontal shift for each prior so that the results under different priors can be distinguished.  
\begin{figure}[ht]
\centerline{\includegraphics[width=5.75in]{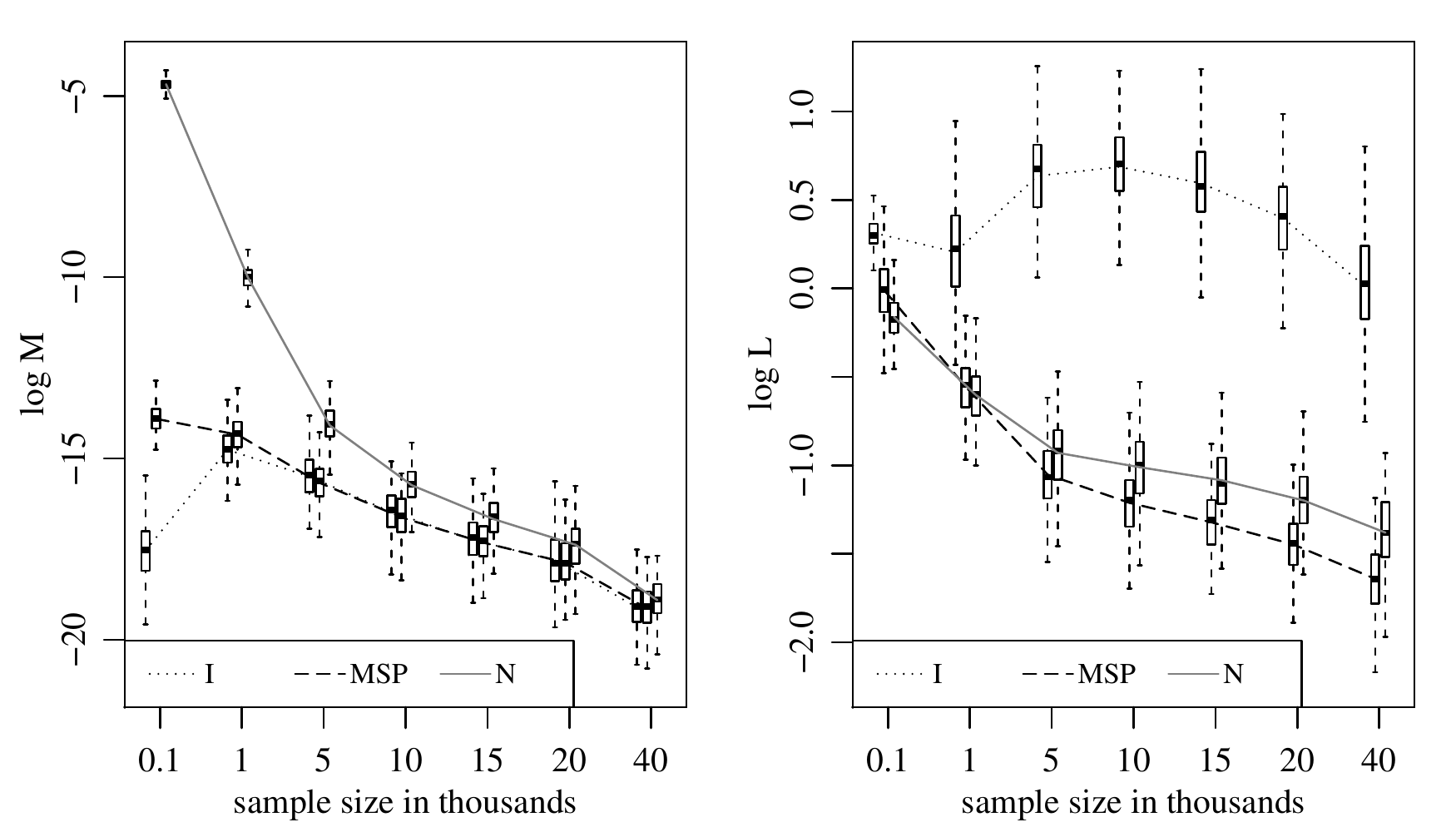} }
\caption{Comparison of $M$ and $L$ metrics  on the log scale for {\PriorOne} (I), {\PriorTwo} (N) and  {\PriorThree} (MSP) at various sample sizes.}
\label{Fi:CTSimStudy}
\end{figure}
Not surprisingly, the estimates of $\theta$  under
{\PriorOne}  and $\pi_1$  outperform  those under $\pi_0^N$, 
as these former two priors were designed to have correct 
prior expectations for $\theta$. 
(The initial non-monotonic trend in the performance of  $\pi_0^I$  
with sample size is due to the fact that $\pi_0^I$ has exactly 
correct prior expectation: If the sample size were zero
then $M$ would be zero as well).
In contrast, the second plot in Figure \ref{Fi:CTSimStudy} 
indicates that $\pi_0^I$ provides relatively poor estimates of the 
dependence functions:
At all sample sizes, this prior  underperforms compared to the other two, 
demonstrating the cost of making {\PriorOne} directly informative about the marginals. 
On the other hand, $\pi_0^N$ and $\pi_1$  have very comparable 
performance in terms of estimation of the dependence functions. 
These comparisons, using both the marginal and margin-free performance metrics,
  highlight the desirable properties of the marginally 
specified prior formulation: A marginally specified prior $\pi_1$ is 
able to represent prior information about specific functionals $\theta(f)$ of 
the high-dimensional parameter $f$ without being overly informative about 
other aspects of the parameter.

 
\section{Discussion}
Nonparametric priors 
for a high-dimensional parameter $f$
based on Dirichlet processes or
Dirichlet distributions do not easily facilitate
partial prior information
about arbitrary functionals $\theta=\theta(f)$. 
Attempts to make such priors informative about $\theta$ can generally 
make the prior undesirably informative about other aspects of $f$. 

In this article, we have presented a 
relatively simple 
 solution to this problem, 
via construction of a marginally specified prior (MSP)
that can induce a target marginal prior on a functional $\theta$, 
but is otherwise as close as possible to a given 
canonical ``noninformative'' nonparametric prior. 
We have provided general posterior approximation schemes for such priors, 
based on conceptually simple modifications to standard MCMC routines 
for canonical nonparametric priors.  In two examples we have shown 
that the MSP behaves as anticipated:
Given accurate prior information, 
the MSP provides 
improved estimation for $\theta$  as compared to 
``noninformative'' priors, while providing similar or better
estimation performance 
for other aspects of the unknown parameter $f$. 

One barrier to the adoption of MSPs is that the posterior approximation 
schemes we have presented require that the ratio $p_1(\theta)/p_0(\theta)$
be computable, where $p_1$ is the desired  informative prior for $\theta$
and $p_0$ is the prior induced on $\theta$ by a canonical prior $\pi_0$. 
Generally, $p_0$ will not have a closed form, and so must be 
approximated numerically or otherwise. If the dimension of $\theta$ is 
small, it should generally be feasible to approximate $p_0$ 
with a kernel density estimate, or by a simple parametric family. 
If $\theta$ is high-dimensional, then other approximation 
strategies may be required, such as approximating the joint density 
of $\theta$ as a product density (i.e.\ assuming independence of 
subvectors of $\theta$) or perhaps by using mixture models. 
While the latter strategy may be more flexible and accurate than the 
former, it may roughly double the modeling efforts in any given problem
by requiring one to essentially 
nonparametrically estimate $p_0$ before estimating $f$. 

Supplementary results and replication code for the material in Section 3
are available at the second author's website: \href{http://www.stat.washington.edu/~hoff}{\nolinkurl{www.stat.washington.edu/~hoff}}

\bibliographystyle{plainnat}
\bibliography{refs}

\end{document}